\documentclass[aps, pra, nofootinbib, 
onecolumn, 11pt, tightenlines,  
notitlepage, longbibliography]{revtex4-1}
\newcommand{\papertitle}{Error rates in quantum circuits}
\newcommand{\pdfauthor}{Joel Wallman}

\usepackage{graphicx, amsmath, amssymb, amsthm, color, bm, bbm, pgfplots,dsfont}
\pgfplotsset{compat=newest}
\usepackage{times} 

\usepackage{hyperref}
\definecolor{darkblue}{RGB}{0,0,127} 
\definecolor{darkgreen}{RGB}{0,150,0}
\hypersetup{breaklinks, colorlinks, linkcolor=darkblue, citecolor=darkgreen, 
	filecolor=red, urlcolor=blue}
\hypersetup{pdftitle={\papertitle}, pdfauthor={\pdfauthor}}

\usepackage{cleveref}

\theoremstyle{plain}
\newtheorem{thm}{Theorem}
\newtheorem{lem}[thm]{Lemma}

\newtheorem{cor}[thm]{Corollary}
\theoremstyle{definition}

\theoremstyle{remark}

\renewenvironment{proof}[1][Proof]{\noindent\textbf{#1.} }{\ $\Box$}

\definecolor{nblue}{rgb}{0.2,0.2,0.7}
\definecolor{ngreen}{rgb}{0.1,0.5,0.1}
\definecolor{nred}{rgb}{0.8,0.2,0.2}
\definecolor{nblack}{rgb}{0,0,0}

\newcommand{\bs}[1]{\boldsymbol{#1}}
\newcommand{\md}[1]{\mathds{#1}}
\newcommand{\mc}[1]{\mathcal{#1}}
\newcommand{\bc}[1]{\boldsymbol{\mathcal{#1}}}
\newcommand{\mbb}[1]{\mathbb{#1}}
\newcommand{\mbf}[1]{\mathbf{#1}}
\newcommand{\ra}{\rangle}
\newcommand{\la}{\langle}
\newcommand{\tn}[1]{^{\otimes #1}}
\newcommand{\ct}{^{\dagger}}

\DeclareMathOperator{\Tr}{\mathrm{Tr}}

\newcommand{\hidden}[1]{}

\usepackage{hyperref}
\definecolor{darkblue}{RGB}{0,0,127} 
\definecolor{darkgreen}{RGB}{0,150,0}
\hypersetup{breaklinks, colorlinks, linkcolor=darkblue, citecolor=darkgreen, filecolor=red, urlcolor=blue}
\hypersetup{pdftitle={\papertitle}, pdfauthor={\pdfauthor}}

\newcommand{\vertii}[1]{{\lVert #1 \rVert}}
\newcommand{\vertiii}[1]{{\left\vert\kern-0.25ex\left\vert\kern-0.25ex\left\vert
 #1 \right\vert\kern-0.25ex\right\vert\kern-0.25ex\right\vert}}

\include{Qcircuit}
\newcommand{\lmeter}{*=<1.8em,1.4em>{\xy 
="j","j"-<.778em,.322em>;{"j"+<.778em,-.322em> \ellipse 
ur,_{}},"j"-<0em,.4em>;p+<.5em,.9em> 
**\dir{-},"j"+<2.2em,2.2em>*{},"j"-<2.2em,2.2em>*{} \endxy} \POS 
="i","i"+UR;"i"+UL **\dir{-};"i"+DL **\dir{-};"i"+DR **\dir{-};"i"+UR 
**\dir{-},"i"}

\begin{document}
\title{\papertitle}

\author{Joel\ \surname{Wallman}}
\affiliation{Institute for Quantum Computing and Department of Applied 
Mathematics, University of Waterloo, Waterloo, Canada}

\date{\today}

\begin{abstract}
Rigorous analyses of errors in quantum circuits, or components thereof, 
typically appeal to the diamond distance, that is, the worst-case error instead 
of the average gate infidelity. This has two primary drawbacks, namely, that 
the diamond distance cannot be directly and efficiently estimated and the 
diamond distance may be unnecessarily pessimistic. In this work, we obtain 
upper- and lower-bounds on the diamond distance that are proportional to each 
other and can be efficiently estimated. We also obtain an essentially 
dimensional-independent upper bound on the average increase in purity due to a 
generalized relaxation process that is proportional to the average gate 
infidelity. We also show that the ``average'' error rate in a generic quantum 
circuit is proportional to the diamond distance, up to dimensional factors. 
Consequently, the average error in a quantum circuit can differ significantly 
from the infidelity, even when coherent processes make a negligible 
contribution to the infidelity.
\end{abstract}

\maketitle

Currently, great experimental effort is being exerted to control quantum 
systems precisely enough to allow for scalable quantum error correction and to 
demonstrate quantum supremacy, that is, perform some task on an experimental 
quantum computer that is not viable on a conventional computer. For both of 
these efforts, it is inadequate to completely characterize every circuit 
component in place, as this would be more difficult then simply simulating a 
quantum computer directly, thus removing any possible computational advantage. 
Rather, the aim of quantum characterization is to provide figures of merit and 
efficient methods of estimating those figures such that a sufficiently small 
figure of merit guarantees that fault-tolerant quantum computation is possible 
and/or that the total error in a circuit is sufficiently small.

There are several different metrics for quantifying the error rate due to some 
noise process $\mc{E}$, which we assume to be a completely positive and 
trace-preserving (CPTP) linear map $\mc{E}:\mbb{H}_d\to\mbb{H}_d$ where 
$\mbb{H}_d$ is the set of $d\times d$ density matrices, that is, Hermitian and 
positive semi-definite matrices with unit trace. The most prominent metrics are 
the average gate infidelity to the identity and the diamond distance from the 
identity (hereafter simply the infidelity and diamond distance 
respectively)~\cite{Fuchs1999}
\begin{align}\label{eq:distances}
r(\mc{E}) 
&= 1 - \int {\rm d}\psi \Tr[\psi\mc{E}(\psi)] \notag\\
\epsilon_\diamond(\mc{E})&= \max_{\rho\in \mbb{H}_{d^2}} \frac{1}{2}\lVert 
(\mc{E}-\mc{I}_d)\otimes\mc{I}_d(\rho) \rVert_1 ,
\end{align}
respectively, where the integral is over the set of pure states according 
to the unitarily-invariant Haar measure $\mathrm{d}\psi$ and
\begin{align}
\|M\|_p = \Bigl(\Tr (M\ct M)^{p/2}\Bigr)^{1/p}
\end{align} 
is the Schatten $p$-norm of $M$. The relevant metric will typically depend on 
the task being performed (and possibly the technique for proving robustness). 
Current methods of proving fault-tolerance utilize the diamond 
distance~\cite{Shor1995,Aharonov1999,Knill1998}. Presently, there 
is no known efficient protocol for directly estimating the diamond distance. 

In contrast, the infidelity can be efficiently estimated using randomized 
benchmarking~\cite{Emerson2005,Dankert2009,Knill2008,Magesan2011,Magesan2012a} 
or direct fidelity estimation~\cite{DaSilva2011,Flammia2011}. If the error 
channel is a stochastic channel (i.e., has a Kraus-operator decomposition into 
trace-orthogonal operators with one operator proportional to the identity), 
then the infidelity provides an exact estimate of the diamond 
distance~\cite{Magesan2012a}. Otherwise, the best-known bound 
is~\cite{Wallman2014}
\begin{align}\label{eq:infidelity_diamond}
\tfrac{d+1}{d}r(\mc{E}) \leq \epsilon_\diamond(\mc{E}) \leq 
\sqrt{d(d+1)r(\mc{E})},
\end{align}
where the individual scalings with respect to $d$ and $r(\mc{E})$ are 
optimal~\cite{Sanders2015}. For target (small) error rates, the upper and lower 
bounds differ by orders of magnitude. There is a belief that, despite the 
dramatically different scaling, the infidelity captures the average 
``computationally relevant'' error~\cite{Knill2008} and that the discrepancy is 
simply a difference between ``average'' and ``worst-case'' performance.

In this paper, we obtain lower- and upper bounds on the diamond distance for 
general noise in terms of functions that can be efficiently estimated, namely, 
the infidelity and the unitarity~\cite{Wallman2015} and differ by essentially a 
factor of $d^{3/2}$. The bounds here improve upon similar bounds obtained 
concurrently in Ref.~\cite{Kueng2015} by a dimensional factor and hold for 
general noise rather than only for unital noise. To obtain these bounds, we 
also prove that a worst-case version of the infidelity is proportional to the 
(average) infidelity and that the relaxation rate (i.e., the average rate at 
which the purity of a state increases) under generalized amplitude processes is 
at most linear in the infidelity, and so does not introduce any significant 
difference between the infidelity and the diamond distance.

Conceptually, the fundamental reason for the disconnect between the diamond 
distance and the infidelity is that the infidelity corresponds to an error rate 
for a measurement in an expected eigenbasis of the system, which is insensitive 
to any coherences between the expected state of the system and the orthogonal 
subspace (quantified by the unitarity). Such coherence is generically 
introduced for some input state by noise that is not completely stochastic, 
which can then have a significant impact on error rates for any measurement 
that is not in a basis containing the expected state. In general applications 
of quantum information, such as quantum computing, the measurement outcome is 
often expected to be nondeterministic even in the absence of noise, that is, 
measurements are often expected to be in a basis that does not contain the 
ideal state of the system. Consequently, the coherent contributions to the 
diamond and induced trace distances will generally be relevant.

We then turn to the problem of estimating the total error rate of a circuit (or 
of a circuit fragment consisting of the gates between two rounds of error 
correction). We define an average error rate and show that it behaves 
qualitatively like the diamond distance. In particular, primary 
contribution of coherent errors to the average error is due to the lowest-order 
terms rather than the possible worst-case alignment of rotations. We also 
demonstrate numerically that there are two characteristic scalings with respect 
to the circuit length $K$, namely, $K$ and $\sqrt{K}$ for systematic and random 
noise respectively. 

\section{The discrepancy between the infidelity and diamond 
distance}\label{sec:diamond_bound}

We now identify the precise cause of significant discrepancies between the 
diamond distance and the infidelity. We first prove that while a worst-case 
version of the infidelity exhibits the same dimensional scaling as the diamond 
distance, it is nevertheless proportional to the average infidelity 
(\cref{cor:worst_case_infidelity}). We then prove a dimension-independent bound 
on the relaxation rate due to a general non-unital process 
(\cref{thm:nonunital_bounds}). Finally, we obtain upper- and lower-bounds on 
the diamond distance that scale as $\sqrt{r}$ instead of $r$ if and only if 
there are significant coherent contributions to the noise as quantified by the 
unitarity~\cite{Wallman2015} (\cref{cor:diamonddistance}).

The infidelity is an average over all pure states. To analyze the dependence on 
the input state, we define the $\psi$-infidelity and max-infidelity to be
\begin{align}
r(\mc{E},\psi) &= 1- \Tr[\psi\mc{E}(\psi)] ,\notag\\
r_{\max}(\mc{E}) &= \max_\psi r(\mc{E},\psi)
\end{align}
respectively, where the average infidelity is $r(\mc{E})=\int {\rm d} \psi\, 
r(\mc{E},\psi)$. In the following, the average $\psi$-infidelity over a basis 
(rather than the sum) is essentially independent of the dimension. We could 
also consider the $\psi$-infidelity of $\mc{E}\otimes \mc{I}$ (allowing for 
entangled inputs), resulting in a change from $d$ to $d^2$ and again giving an 
average infidelity over the basis that is also essentially independent of the 
dimension.

\begin{thm}\label{thm:psi_infidelity}
For any $d\in\mbb{N}$, orthonormal basis $\{|j\ra:j\in\mbb{Z}_d\}$ of 
$\mbb{C}^d$, and CPTP map $\mc{E}$,
\begin{align}\label{eq:basis_infidelity}
\sum_{j\in\mbb{Z}_d} r(\mc{E},|j\ra\!\la j|) \leq (d+1)r(\mc{E}).
\end{align}
\end{thm}

\begin{proof}
Fix $d\in\mbb{N}$, $\{|j\ra:j\in\mbb{Z}_d\}$ to be an arbitrary orthonormal 
basis of $\mbb{C}^d$, and define the Heisenberg-Weyl operators
\begin{align}
X_d &= \sum_{j\in\mbb{Z}_d} |j\oplus 1\ra\!\la j| \notag\\
Z_d &= \sum_{j\in\mbb{Z}_d} \omega_d^j|j\ra\!\la j| \notag\\
W_{a,b} &= X^a Z^b
\end{align}
with respect to this basis, where $\oplus$ denotes addition modulo $d$ and 
$\omega_d = \exp(2\pi i/d)$. As $W_{a,b}|j\ra\!\la j|W_{a,b}\ct = |j\oplus 
a\ra\!\la j\oplus a|$ and the trace is linear,
\begin{align}\label{eq:equivalent_twirl}
\sum_{j\in\mbb{Z}_d} r(\mc{E},|j\ra\!\la j|) = \sum_{j\in\mbb{Z}_d} 
r(\mc{T}_W[\mc{E}],|j\ra\!\la j|)
\end{align}
where 
\begin{align}
\mc{T}_W[\mc{E}] = \frac{1}{d^2}\sum_{a,b\in\mbb{Z}_d} 
\mc{W}_{a,b}\ct\mc{E}\mc{W}_{a,b},
\end{align}
and $r(\mc{E}) = r[\mc{T}_W(\mc{E})]$. The Heisenberg-Weyl operators satisfy
\begin{align}
\omega_d^{-bc} W_{a,b} W_{c,d} = \omega_d^{-ad} W_{c,d}W_{a,b}
\end{align}
and so $\mc{T}_W(\mc{E})$ is Weyl-covariant, that is,
\begin{align}
\mc{W}_{a,b}\mc{T}_W(\mc{E}) = \mc{T}_W(\mc{E}\mc{W}_{a,b}).
\end{align}
Therefore there exists some probability distribution $p_{a,b}$ over 
$\mbb{Z}_d^2$ such that~\cite{Holevo2005}
\begin{align}\label{eq:weyl_twirl}
\mc{T}_W(\mc{E})(\rho) = \sum_{a,b\in\mbb{Z}_d} p_{a,b} W_{a,b} \rho W_{a,b}\ct
\end{align}
and, as the infidelity is linear and weakly unitarily invariant,
\begin{align}
r(\mc{E}) = r(\mc{T}_W\mc{E}) &= 1-\frac{\sum_{a,b\in\mbb{Z}_d} \lvert 
\Tr\sqrt{p_{a,b}} 
	W_{a,b}\rvert^2 +d}{d^2+d} \notag\\
&= \frac{d (1-p_{0,0})}{d+1}.
\end{align}
Substituting \cref{eq:weyl_twirl} into the right-hand-side of 
\cref{eq:equivalent_twirl} and using $p_{a,b}\geq 0$ gives
\begin{align}
\sum_{j\in\mbb{Z}_d} 
r(\mc{T}_W[\mc{E}],|j\ra\!\la j|) &= d(1-\sum_{b\in\mbb{Z}_d} p_{0,b}) \notag\\
&\leq (d+1)r(\mc{E}).
\end{align}
Furthermore, this bound is saturated by
\begin{align}
\mc{E}'(\rho) = p\rho + (1-p)X\rho X\ct.
\end{align}
\end{proof}

We now prove that the dimensional scaling from \cref{thm:psi_infidelity} is 
optimal. Consequently, the worst-case infidelity can exhibit the same 
dimensional scaling as the diamond norm but is always proportional to the 
average infidelity.

\begin{cor}\label{cor:worst_case_infidelity}
For any CPTP map $\mc{E}$, the worst-case infidelity satisfies
\begin{align}
r(\mc{E}) \leq \max_\psi r(\mc{E},\psi) \leq (d+1)r(\mc{E}).
\end{align}
Furthermore, there exist CPTP maps such that
\begin{align}\label{eq:infidelity_lowerbound}
\max_\psi r(\mc{E},\psi) \geq \frac{(d^2+d)r(\mc{E})}{4(d-1)}.
\end{align}
\end{cor}

\begin{proof}
The upper bound follows from \cref{thm:psi_infidelity} with the non-negativity 
of the $\psi$-infidelity. \Cref{eq:infidelity_lowerbound} can be obtained by 
setting $\mc{E} = \mc{U}_\phi$ where $U_\phi = \md{I}_d + 
(e^{i\theta}-1)|0\ra\!\la0|$, which has average infidelity
\begin{align}
r(\mc{U}_\phi) = 1-\frac{\lvert \Tr U_\phi\rvert^2 +d}{d^2+d} = 
\frac{2(d-1)(1-\cos\theta)}{d^2+d}.
\end{align}
Evaluating the $\psi$-infidelity for $|+\ra:= (|0\ra + |1\ra)\sqrt{2}$ gives
\begin{align}
\max_\psi r(\mc{U}_\phi,\psi) \geq r(\mc{U}_\phi,|+\ra\la+|) &= 
\frac{1}{2}(1-\cos\theta) \notag\\
&= \frac{(d^2+d)r(\mc{U}_\phi)}{4(d-1)}.
\end{align}
\end{proof}

\Cref{cor:worst_case_infidelity} demonstrates that the potential discrepancy 
between the $O(\sqrt{r})$ scaling of the diamond distance and the average 
infidelity is not simply the difference between ``average'' and ``worst-case''. 
The discrepancy only arises for channels that are not stochastic. The primary 
examples of non-stochastic processes are relaxation and coherent processes.

Relaxation to a ground state is a common physical process that cannot be 
described as a stochastic channel. The canonical example is the single-qubit 
amplitude damping channel with Kraus operators
\begin{align}
K_1 = |0\ra\!\la0| + \sqrt{1-\gamma}|1\ra\!\la1|,\ K_2 = 
\sqrt{\gamma}|0\ra\!\la1|.
\end{align}
For multiple qubits, the energy eigenbasis will generically be entangled 
relative to the computational basis. We now prove that the relaxation rate 
$\alpha(\mc{E})$ (defined below, which quantifies the average increase in 
purity) is proportional to the infidelity and essentially independent of the 
dimension. Note that the following bound coincides with \cite[Prop. 
12]{Wallman2014} for $d=2$. 

\begin{thm}\label{thm:nonunital_bounds}
For any CPTP map $\mathcal{E}\in\mbb{T}_d$, 
\begin{align}
\alpha(\mc{E}):=\|\mc{E}(\tfrac{1}{d}\md{I}_d)-\tfrac{1}{d}\md{I}_d\|_2 
\leq  
\frac{\sqrt{2}(d+1)r(\mc{E})}{d}.
\end{align}
\end{thm}

\begin{proof}
Consider the spectral resolution $\mc{E}(\md{I}_d) = \sum_{j,k} c_{jk} 
|j\ra\!\la j|$ 
in some orthonormal basis $\{|j\ra:j\in\mathbb{Z}_d\}$ of $\mathbb{C}^d$, where 
\begin{align}
c_{j,k} = \Tr |j\ra\!\la j|\mc{E}(|k\ra\!\la k|) \in[0,1].
\end{align}
Then
\begin{align}
\alpha(\mc{E})^2
&= \Tr \mc{E}(\tfrac{1}{d}\md{I}_d)^2 - 
\frac{2}{d}\Tr\mc{E}(\tfrac{1}{d}\md{I}_d) + \frac{1}{d} \notag\\
&= \frac{1}{d^2}\sum_j \bigl(\sum_k c_{jk}\bigr)^2 -\frac{1}{d} \notag\\
&= \frac{1}{d^2}\sum_j c_{j,j}^2 + \frac{1}{d^2}\sum_j 2c_{j,j}v_j + 
\frac{1}{d^2}\sum_j v_j^2 - \frac{1}{d} \notag\\
&= \frac{1}{d^2}\sum_j (r(\mc{E},|j\ra\!\la j|)-v_j)^2
\end{align}
as the basis is orthonormal with $v_j = \sum_{k\neq j} c_{j,k}$ and we have 
used $c_{j,j} = 1 - r(\mc{E},|j\ra\!\la j|)$ and
\begin{align}
\sum_j v_j &= \sum_{j,k\neq j} c_{j,k} \notag\\
&= \sum_{k,j\neq k} c_{j,k} \notag\\
&= \sum_j r(\mc{E},|j\ra\!\la j|) 
\end{align}
for trace-preserving maps. The difference $r(\mc{E},|j\ra\!\la j|)-v_j = \Tr 
(\md{I}_d - |j\ra\!\la j|)\mc{E}(|j\ra\la j|) - \Tr |j\ra\la j|\mc{E}(\md{I}_d 
- |j\ra\!\la j|)$ is the net flow out of the state $|j\ra\!\la j|$ and will 
typically be smaller than indicated by the following bound. Noting that 
$v_j,r(\mc{E},|j\ra\!\la j|)\geq 0$, we have
\begin{align}
\alpha(\mc{E})^2 &\leq \frac{1}{d^2}\sum_j [r(\mc{E},|j\ra\!\la j|)^2 + 
v_j^2] 
\notag\\
&\leq \frac{1}{d^2}\Bigl(\sum_j r(\mc{E},|j\ra\!\la j|\Bigr)^2 + 
\frac{1}{d^2}\Bigl(\sum_j v_j\Bigr)^2 \notag\\
&= \frac{2}{d^2}\Bigl(\sum_j r(\mc{E},|j\ra\!\la j|)\Bigr)^2 \notag\\
&\leq 2d^{-2}(d+1)^2 r(\mc{E})^2,
\end{align}
where the final inequality follows from \cref{thm:psi_infidelity}.
\end{proof}

We now obtain lower- and upper-bounds on the diamond norm of a general linear 
map $\mc{T}$ in terms of the purity of the Jamio{\l}kowski-isomorphic state 
\begin{align}
J(\mc{T}) &= d^{-1}\sum_{j,k\in\mbb{Z}_d} \mc{T}(|j\ra\!\la k|)\otimes 
|j\ra\!\la k|
\end{align}
that differ by a factor of $d^{3/2}$. We will then state upper and lower bounds 
on the diamond distance in terms of the infidelity and the 
unitarity~\cite{Wallman2015}
\begin{align}
u(\mc{E}) = \frac{d}{d-1}\int\mathrm{d}\psi\, \Tr\mc{E}(\psi-\md{I}_d/d)^2,
\end{align}
both of which can be efficiently estimated experimentally. In particular, 
$u(\mc{E})\in[p(\mc{E})^2,1]$ with $u-p(\mc{E})^2 \in O(r^2)$ giving a diamond 
distance that scales linearly with $r(\mc{E})$ if $\mc{E}$ is stochastic where 
$p(\mc{E})=1-dr(\mc{E})/(d-1)$ is the randomized benchmarking decay constant.

\begin{thm}\label{thm:diamond_bounds}
For any linear map $\mc{T}:\mbb{C}^{d\times 
	d}\to\mbb{C}^{d\times d}$,
\begin{align}
\lVert J(\mc{T})\rVert_2 \leq
\vertiii{\mc{T}}_1
\leq d^{3/2}\lVert J(\mc{T})\rVert_2
\end{align}
where
\begin{align}
\vertiii{\mc{T}}_1 = \sup_{A\in\mbb{C}^{d^2\times d^2}:\vertii{A}_1=1} 
\vertii{\mc{T}\otimes\mc{I}_d(A)}_1 .
\end{align}
\end{thm}

\begin{proof}
By \cite[Thm. 6]{Watrous2012},
\begin{align}
\vertiii{\mc{T}}_1 = \sup_{\rho,\sigma\in\in\mbb{H}_d} 
d\vertii{(\md{I}_d\otimes\sqrt{\rho})J(\mc{T})(\md{I}_d\otimes\sqrt{\sigma})}_1 
\end{align}
noting that our convention of $J(\mc{T})$ differs from Ref.~\cite{Watrous2012} 
by a factor of $d$. The lower bound can be obtained by setting 
$\rho=\sigma=\md{I}_d/d$ and applying \cref{eq:standardCS}. By H{\"{o}}lder's 
inequality,
\begin{align}
\lVert A B C \rVert_1 
\leq \lVert A\rVert_\infty \lVert BC\rVert_1 
\leq  \lVert A\rVert_\infty \lVert B\rVert_2 \lVert C\rVert_2
\end{align}
for arbitrary $A,B,C\in\mbb{C}^{m\times m}$. As $\rho,\sigma\in\mbb{H}_d$, they 
are positive semi-definite and have unit trace and so $\lVert 
\md{I}_d\otimes\sqrt{\rho}\rVert_\infty \leq 1$ and $\lVert 
\md{I}_d\otimes\sqrt{\sigma}\rVert_2 \leq \sqrt{d}$ for all 
$\rho,\sigma\in\mbb{H}_d$. (Note that the bound in Ref.~\cite{Wallman2014} uses 
$\lVert BC\rVert_1 \leq \lVert B\rVert_1 \lVert C\rVert_\infty$ instead, which, 
together with the Fuchs-van de Graaf inequality~\cite{Fuchs1999}, gives the 
bound in \cref{eq:infidelity_diamond}.)
\end{proof}

\begin{cor}\label{cor:diamonddistance}
For any quantum channel $\mc{E}\in\mbb{T}_d$,
\begin{align}
\frac{1}{\sqrt{2}}\lVert J(\Delta)\rVert_2 \leq
\epsilon_\diamond(\mc{E})
\leq \frac{d^{3/2}}{2}\lVert J(\Delta)\rVert_2
\end{align}
where $\Delta = \mc{E}-\mc{I}$. In terms of the infidelity and unitarity,
\begin{align}
\frac{C}{\sqrt{2}}
\leq \epsilon_\diamond(\mc{E})
\leq \sqrt{\frac{d^3 C^2}{4} + \frac{(d+1)^2 r(\mc{E})^2}{2}},
\end{align}
where
\begin{align}
C^2 = \frac{d^2-1}{d^2} (u(\mc{E})-2p(\mc{E})+1).
\end{align}
\end{cor}

\begin{proof}
First note the factor of $1/2$ from the definition of 
$\epsilon_\diamond(\mc{E})$ in \cref{eq:infidelity_diamond} and we obtain the 
$\sqrt{2}$ improvement from \cref{lem:sharperCS} for trace-preserving maps 
$\mc{E}$, so that $\Tr\Delta(A)=0$ for all $A$. By \cite[Prop. 9]{Wallman2015},
\begin{align}
\lVert J(\Delta)\rVert_2^2 
&= \frac{d^2-1}{d^2} u(\Delta) + \frac{\alpha(\mc{E})^2}{d} \notag\\
&= \frac{d^2-1}{d^2} (u(\mc{E})-2p+1) + \frac{\alpha(\mc{E})^2}{d}
\end{align}
where we have used $\Tr \Delta(A)=0$ for all $A$ as $\mc{E}$ is a 
trace-preserving map, and $d\alpha(\mc{E})^2 = \lVert 
\bc{E}_{\mathrm{n}}\rVert_2^2$. The lower and upper bounds follow from the 
non-negativity of norms and \cref{thm:nonunital_bounds} respectively.
\end{proof}

\begin{lem}\label{lem:sharperCS}
For any traceless Hermitian matrix $M\in\mbb{C}^{d\times d}$,
\begin{align}\label{eq:CS_singular}
\sqrt{2}\|M\|_2 \leq \|M\|_1 \leq \sqrt{d}\|M\|_2.
\end{align}
Moreover, both these bounds are saturated.
\end{lem}

\begin{proof}
For any Hermitian matrix $M\in {\rm GL}(d)$, writing $M = U\bs{\eta}U\ct$ 
where $\bs{\eta}$ is a diagonal matrix whose entries are the eigenvalues 
$\{\eta_j\}$ of $M$ and using the unitary invariance of the trace and Frobenius 
norms gives the standard bounds
\begin{align}\label{eq:standardCS}
\|M\|_2= \sqrt{\sum_{j=1}^d \eta_j^2}
\leq \|M\|_1 
= \sum_{j=1}^d |\eta_j|
\leq \sqrt{d\sum_{j=1}^d \eta_j^2} 
= \sqrt{d}\|M\|_2
\end{align}
by the Cauchy-Schwarz inequality. 

To obtain a sharper lower bound for traceless Hermitian matrices, let 
\begin{align}
\bs{\eta} = \bs{\eta}^+ \oplus -\bs{\eta}^-
= \left(\begin{array}{cc} \bs{\eta}^+ & 0 \\ 0 & -\bs{\eta}^- 
\end{array}\right)
\end{align}
where $\oplus$ denotes the matrix direct sum and $\eta^{\pm}$ are both 
positive 
semidefinite. Then
\begin{align}
\|M\|_2
= \sqrt{\sum_{j=1}^d \eta_j^2}
= \sqrt{\|\bs{\eta}^+\|_2^2 + \|\bs{\eta}^-\|_2^2}
\leq \sqrt{\|\bs{\eta}^+\|_1^2 + \|\bs{\eta}^-\|_1^2}
= \sqrt{\frac{\|\bs{\eta}\|_1^2}{2}}
= \frac{1}{\sqrt{2}}\|M\|_1,
\end{align}
where we have used $\|\bs{\eta}^{\pm}\|_1 = \tfrac{1}{2}\|\bs{\eta}\|_1$ 
which holds for traceless matrices. The above lower bound is saturated when 
$\eta_{11}=1$, $\eta_{22}=-1$ and all other eigenvalues are zero.
\end{proof}

\section{The error per gate cycle in a quantum 
circuit}\label{sec:computationallyrelevant}

\begin{figure}
\begin{tabular}{c c c c c c c}
\Qcircuit @C=1em @R=.7em {
a)\hspace{10mm} & \lmeter & \gate{G_K} & \qw & \qw & \ldots & & 
\gate{G_1} & \qw & \rstick{\psi}\qw \\
b)\hspace{10mm} & \lmeter & \gate{G_K} & \gate{\mathcal{E}_K} & \qw & 
\ldots & & \gate{G_1} & \gate{\mathcal{E}_1} & \rstick{\psi}\qw
}
\end{tabular}
\caption{a) A quantum circuit $\mathcal{C}$ with $K$ rounds of gates, where 
time goes from right to left. b) An implementation $\tilde{\mathcal{C}}$ of 
$\mathcal{C}$ with Markovian noise.}\label{fig:qcircuit}
\end{figure}
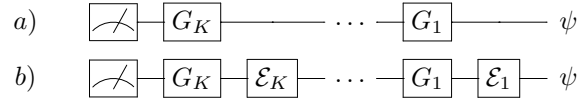

We now identify the ``computationally relevant'' error per gate cycle under 
Markovian noise and prove that it is closer to the diamond distance than the 
infidelity (\cref{thm:average_diamond}) and illustrate the different scalings 
with respect to the infidelity, the circuit length and the diamond distance.

A quantum circuit $\mathcal{C}$ consists of some $d$-dimensional input state 
$\psi$ (e.g., an $n$-qubit computational basis state), $K$ rounds of gates 
$G_1$, $\ldots$, $G_K$, and a final measurement $\mbf{M}$ (e.g., of a single 
qubit in the computational basis) as in \cref{fig:qcircuit}~(a). Under the 
assumption of Markovian noise, an experimental implementation 
$\tilde{\mathcal{C}}$ of $\mathcal{C}$ consists of a preparation of 
$\psi$~\footnote{The preparation of a mixed state corresponds to 
preparing a pure state and then applying some CPTP map, which can then be 
incorporated into the noise in the first gate.}, $K$ rounds of noisy gates 
$\tilde{\mc{G}}_1$, $\ldots$, $\tilde{\mc{G}}_K$ such that $\tilde{\mc{G}}_k = 
\mc{G}_k\mc{E}_k$ for some CPTP maps $\{\mc{E}_k\}$, and a final measurement 
$\tilde{\mbf{M}}\approx\mbf{M}$ as in \cref{fig:qcircuit}~(b), where we use 
calligraphic font to denote CPTP maps with $\mc{U}(A) = UAU\ct$ for any unitary 
operator $U$ and we denote the composition of channels by a product. For 
clarity of notation, we define 
\begin{align}
A_{b:a} = \begin{cases}
A_b\ldots A_a & \mbox{if } b\geq a \\
\md{I} & \mbox{otherwise.}
\end{cases}
\end{align} 

For decision or function problems (including, e.g., Shor's 
algorithm~\cite{Shor1999}), we can coarse-grain $\mbf{M}$ into two elements, 
$M_0$ and $\md{I}-M_0$, such that outcomes in $M_0$ ($\md{I} -M_0$) give a 
correct (incorrect) answer to the computation. The error rate of the 
experimental implementation $\hat{\mathcal{C}}$ (which does not include the 
error rate of the ideal circuit $\mc{C}$) is then
\begin{align}\label{eq:circuit_error}
\tau(\mathcal{C},\hat{\mathcal{C}}) &= \frac{1}{2}\sum_j \lvert 
\langle \tilde{M}_j, \tilde{\mc{G}}_{K:1}(\psi)\rangle
- \langle M_j,\mc{G}_{K:1}(\psi)\rangle  \rvert \notag\\
&=\lvert \langle \tilde{M}_0, \tilde{\mc{G}}_{K:1}(\psi)\rangle
- \langle M_0,\mc{G}_{K:1}(\psi)\rangle \rvert \notag\\
&=\lvert \langle \tilde{M}_0 - M_0, \tilde{\mc{G}}_{K:1}(\psi)\rangle 
+ \sum_{k=1}^K \langle M_0, 
\mc{G}_{K:k}(\mc{I}-\mc{E}_k)\tilde{\mc{G}}_{k-1:1}(\psi) 
\rangle
\rvert
\end{align}
for CPTP maps, where we have used
\begin{align}
	A_{b:a}-B_{b:a} = \sum_{k=a}^b A_{b:k+1}(A_k-B_k)B_{k-1:1}.
\end{align}
As we are primarily interested in the error due to gates, we set $\tilde{M}_0 = 
M_0$ (although we can use the triangle inequality to separate this out as 
necessary).

For a non-trivial circuit, that is, $G_{k:1}\neq \md{I}_d$, $M_k$ ($\approx 
\tilde{M}_k$) and $\psi_k$ will be in different bases that are fixed relative 
to each other and so the terms in the summation will typically be close to
\begin{align}\label{eq:mean_error}
	t_{\mathrm{avg},m}(\mathcal{E}) = 
	\md{E}_{\psi,\mbf{M}|\Tr M=m,M^2=M} \lvert\langle 
	M,\mathcal{E}(\psi)-\psi\rangle\rvert,
\end{align}
where $m = \Tr M_0$. We now prove that this quantity is proportional to the 
diamond distance, up to a dimensional factor and a factor that depends on the 
rank of $M_0$. For many problems, such as Shor's factoring algorithm, and for 
the measurements between rounds of error correction, $m\in O(d^\alpha)$ and so 
the dimensional scaling between the following upper and lower bounds is 
primarily a consequence of the corresponding dimensional factors in 
\cref{cor:diamonddistance}.

\begin{thm}\label{thm:average_diamond}
	For any quantum channel $\mc{E}$ and $m\in\mbb{N}$,
\begin{align}
\frac{2m(m+1)\epsilon_\diamond(\mc{E})}{d^3(d+1)^2} \leq 
t_{\mathrm{avg},m}(\mc{E}) 
\leq 2\epsilon_\diamond(\mc{E})
\end{align}
where $\Delta = \mc{E} - \mc{I}_d$.
\end{thm}

\begin{proof}
The upper bound is trivial from the definition of $\epsilon_\diamond(\mc{E})$ 
and \cref{cor:diamonddistance} (although note the factor of 2). For the lower 
bound, let 
\begin{align}\label{eq:comp_error}
t(M,\mathcal{E},\psi) = \lvert\langle M,\Delta(\psi)\rangle\rvert,
\end{align}
so that $t_{\mathrm{avg},m}(\mc{E}) = \md{E}_{\psi,M|\Tr M=m} 
t(M,\mc{E},\psi)$. The absolute value in \cref{eq:comp_error} makes evaluating 
the mean difficult. To circumvent this, we use the identity $a^2 = a\otimes a$ 
for $a\in\mbb{R}^d$ and the distributivity of the tensor product to obtain
\begin{align}\label{eq:variance}
\md{V}&= \md{E}_{\psi,M|\Tr M=m,M^2=M} [t(M,\mc{E},\psi)^2] \notag\\
&= \int_{\mbb{S}_d}\mathrm{d}\psi 
\int_{\mathrm{U}(d)}\mathrm{d}U\, 
\Tr\left[\mc{U}(M^{(m)})\tn{2} \Delta\tn{2}
(\psi\tn{2})\right] .
\end{align}
To evaluate the integrals, let $S = \sum_{i,j} |ij\ra\!\la ji|$ be the 
two-qudit swap gate and let $\pi_{a/s} = (\md{I}_{d^2} \pm S_d)/2$ be the 
projectors onto the symmetric and antisymmetric subspaces of $\mbb{C}^{d^2}$ 
respectively. For any Hermitian matrix $M$,
\begin{align}
\int \mathrm{d}U \mc{U}\tn{2}(M) = \frac{\Tr \pi_s M}{\Tr \pi_s}\pi_s + 
\frac{\Tr \pi_a M}{\Tr \pi_a}\pi_a
\end{align}
by Schur-Weyl duality and so, with $\Tr S(A\otimes B) = \Tr AB$,
\begin{align}
	\int \mathrm{d}\psi\,\psi\tn{2} &= \int \mathrm{d}U\,\mc{U}(\phi\tn{2}) 
	\notag\\
	&= 	\frac{\md{I}_{d^2} + S}{d^2+d} \notag\\
	\int_{\mathrm{U}(d)}\mathrm{d}U\,\mc{U}(M^{(m)})\tn{2} &= 
	\frac{2m^2+2m}{d^2+d}\pi_s + \frac{2m^2 - 2m}{d^2 - d}\pi_a
\end{align}
where $\phi$ is an arbitrary pure state. Substituting this into 
\cref{eq:variance} gives
\begin{align}
\md{V} &= \frac{m^2 + m}{(d^2+d)^2}
\Tr\left[(\md{I}_{d^2}+S) \Delta\tn{2}(\md{I}_{d^2} + S)\right] \notag\\
&= \frac{m^2 + m}{(d^2+d)^2}
\Tr\left[S \Delta\tn{2}(\md{I}_{d^2} + S)\right] \notag\\
&=\frac{m^2 + m}{(d+1)^2}(\|J(\Delta)\|_2^2 + \|\Delta(\md{I}_d)\|_2^2) \notag\\
&\geq \frac{4m(m+1)\epsilon_\diamond(\mc{E})^2}{d^3(d+1)^2}
\end{align}
where we have used $\pi_a \Delta\tn{2}(\pi_s)=0$, \cite[Prop. 
3]{Wallman2015}, the non-negativity of norms, and \cref{cor:diamonddistance}. 
Now note that
\begin{align}
\max t(M,\mathcal{E},\psi) \leq 2\epsilon_\diamond(\mc{E}),
\end{align}
and so, with \cref{cor:diamonddistance} and the non-negativity of norms,
\begin{align}
t_{\rm avg}(\mc{E}) 
&\geq \frac{\md{E}_{\psi,\mbf{M}}[t(M,\mc{E},\psi)^2]}{\max 
t(M,\mc{E},\psi)}\notag\\
&\geq \frac{2m(m+1)\epsilon_\diamond(\mc{E})}{d^3(d+1)^2}.
\end{align}
\end{proof}

\Cref{thm:average_diamond} gives a lower bound on the typical size of the 
individual terms in \cref{eq:circuit_error}, which scales as $\sqrt{r}$ for 
general noise processes. However, these individual terms can have arbitrary 
signs and so can combine in several ways. In particular, there will be two 
characteristic scalings with respect to $K$, namely, 
$\tau(\mc{C},\hat{\mc{C}})\in O(K)$ when all the signs are the same and 
$\tau(\mc{C},\hat{\mc{C}})\in O(\sqrt{K})$ when the signs are essentially 
uniformly random.

Almost all the signs will be the same in at least two regimes as illustrated in 
\cref{fig:circuit_errors}~(a), namely, for systematic stochastic and unitary 
noise. However, for general stochastic and unitary noise, the signs will be 
uniformly random and so the summation in \cref{eq:circuit_error} will look more 
like a one-dimensional random walk with typical step size proportional to 
$t_{\rm avg,m}(\mc{E})$ and so the total error will scale as $\sqrt{K}t_{\rm 
avg,m}(\mc{E})$ as illustrated in \cref{fig:circuit_errors}~(b).

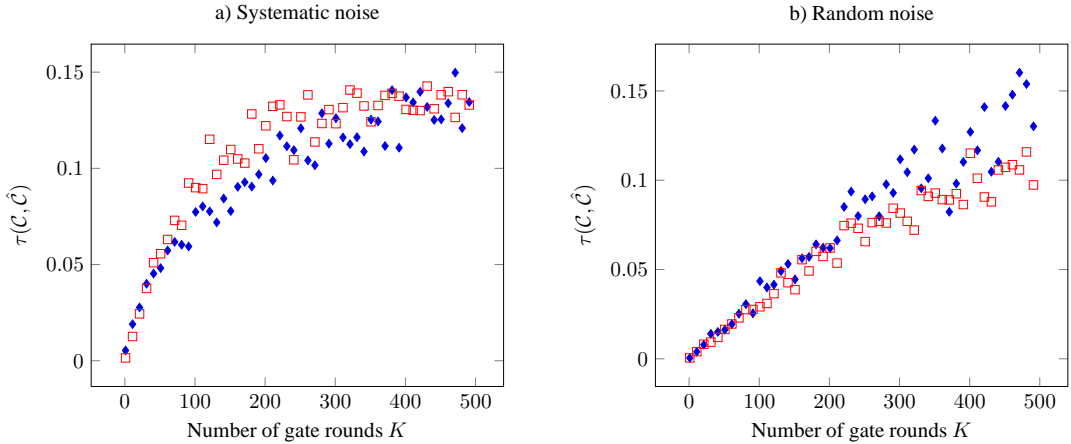
\begin{figure}
\begin{tikzpicture}[scale=.8]
\begin{axis}[
xlabel={Number of gate rounds $K$}, 
ylabel={$\tau(\mc{C},\hat{\mc{C}})$},
yticklabel style={/pgf/number format/fixed},
title = {a) Systematic noise},
]
\addplot+[mark = diamond*, only marks, ]
table[x=numgates, y=mean, col sep=comma,] 
{sumunitarynoise.csv};
\addplot+[mark = square, only marks, ]
table[x=numgates, y=mean, col sep=comma,] 
{sumstochasticnoise.csv};
\end{axis}
\end{tikzpicture}
\qquad
\begin{tikzpicture}[scale=.8]
\begin{axis}[
xlabel={Number of gate rounds $K$}, 
ylabel={$\tau(\mc{C},\hat{\mc{C}})$},
yticklabel style={/pgf/number format/fixed},
title = {b) Random noise},
]
\addplot+[mark = diamond*, only marks, ]
table[x=numgates, y=mean, col sep=comma,] 
{worstsumunitarynoise.csv};
\addplot+[mark = square, only marks, ]
table[x=numgates, y=mean, col sep=comma,] 
{worstsumstochasticnoise.csv};
\end{axis}
\end{tikzpicture}
\caption{(Color online) Average error $\tau(\mc{C},\hat{\mc{C}})$ from 
\cref{eq:circuit_error} averaged over 100 random three-qubit circuits of each 
length $K$ under unitary (blue diamonds) and stochastic (red squares) noise. 
The ideal circuits $\mc{C}$ consist of preparations in the state 
$|\psi\ra\tn{3}$, an $X$-basis measurement on the final qubit and $K$ gates 
drawn uniformly from the set of all combinations of: a) 
controlled-phase, and $R_{\pi/8} = |0\ra\!\la 0|+\exp(\pi i/4)|1\ra\!\la 1|$ 
gates; and b) controlled-phase, Hadamard and $R_{\pi/8}$ gates acting 
disjointly on three qubits. The systematic unitary noise is $U = 
\cos(\theta)\md{I}_8 + i\sin(\theta)ZIZ$ with $\theta = 0.001$ and the 
systematic stochastic noise is $\mc{E}(\rho) = 0.999\rho + 0.001 ZIZ\rho ZIZ$. 
The random noise is fixed and gate-independent for each circuit. The random 
stochastic noise is generated by randomly generating a Pauli channel with 
infidelity 2/300 and the random unitary noise is generated by conjugating $U = 
\cos(\theta)\md{I}_8 + i\sin(\theta)ZIZ$ by a Haar-random unitary with $\theta 
= 0.03$.
\label{fig:circuit_errors}}
\end{figure}

\section{Conclusion}

We have obtained improved bounds on the diamond distance in terms of 
the infidelity and the unitarity~\cite{Wallman2015}, which can both be 
efficiently estimated. When noise is approximately stochastic, the improved 
bound scales as $O(r)$. If the unitarity indicates that the noise contains 
coherent errors, then the improved upper and lower bounds both scale as 
$O(\sqrt{r})$. 

We have also shown that the diamond distance and infidelity are not 
``worst-case'' and ``average'' error rates in the same sense by constructing 
worst-case and average versions of the infidelity and diamond distance 
respectively and showing that they are proportional to the standard versions up 
to dimensional factors. An important semantic consequence is that referring to 
the diamond distance as the worst-case error rate and the infidelity as the 
average error rate is misleading because they quantify error rates in 
fundamentally different ways.

We then provided analytic and numerical arguments demonstrating that the total 
error rate of a circuit (or of a circuit fragment consisting of the gates 
between two rounds of error correction) behaves qualitatively more like the 
diamond distance instead of the infidelity. The primary contribution of 
coherent errors to the average error is not a possible worst-case alignment of 
rotations giving an infidelity that scales quadratically with the number of 
gates. Rather, the primary contribution is due to the lowest order terms 
behaving like the diamond distance as the effective preparations and 
measurements (when propagated through the circuit to immediately precede and 
succeed a lowest-order error) are not in the same basis unless the circuit is 
trivial. As noted in Ref.~\cite{Wallman2015a}, the $\sqrt{r}$ scaling can be 
dominant even when any coherent noise processes make a negligible contribution 
(e.g., 1\%) to the infidelity.

We also demonstrated numerically that there are two characteristic
scalings with the circuit length $K$, namely, $K$ and $\sqrt{K}$ for systematic 
and random noise respectively. In particular, the average error does not 
demonstrate $K^2$ scaling even for systematic coherent errors as would be 
expected if the average error was directly proportional to the infidelity.

\acknowledgments 

The author acknowledges helpful discussions with Arnaud Carignan-Dugas, Hillary 
Dawkins, Joseph Emerson, Steve Flammia and Yuval Sanders. This research was 
supported by the U.S. Army Research Office through grant W911NF-14-1-0103.

\bibliography{library}

\end{document}